\documentclass[letterpaper]{article}
\usepackage{times}
\usepackage{helvet}
\usepackage{courier}
\usepackage{graphicx}
\usepackage{comment}
\usepackage{color}
\usepackage{authblk}
\usepackage{amsthm}
\usepackage{amsmath}
\usepackage{amsfonts}

\newtheorem{definition}{Definition}
\newtheorem{theorem}[definition]{Theorem}

\newtheorem{proposition}[definition]{Proposition}

\newcommand{\wsl}[1]{{\small\color{blue}{\bf [*** WSL: #1]}}}
\newcommand{\cmh}[1]{{\small\color{red}{\bf [*** CMH: #1]}}}

\begin{document}
\title{Tuning the Diversity of Open-Ended Responses From the Crowd}
  
\author[1]{Walter S. Lasecki  \thanks{Supported in part by a Microsoft Research Graduate Fellowship.}}
\author[2]{Christopher M. Homan \thanks{Supported in part by NSF Award \#SBE-1111016.}}
\author[3]{Jeffrey P. Bigham \thanks{Supported in part by NSF Award \#IIS-1149709.}}
\affil[1]{University of Rochester\\Computer Science\\Rochester, NY 14627, wlasecki@cs.rochester.edu}
\affil[2]{Rochester Institute of Technology\\Computer Science\\Rochester, NY 14623, cmh@cs.rit.edu}
\affil[3]{Carnegie Mellon University\\HCI and LTI Institutes\\Pittsburgh, PA 15213, jbigham@cs.cmu.edu}

\maketitle

\begin{abstract}
Crowdsourcing can solve problems that current fully automated systems cannot. Its effectiveness depends on the reliability, accuracy, and speed of the crowd workers that drive it. These objectives are frequently at odds with one another. For instance, how much time should workers be given to discover and propose new solutions versus deliberate over those currently proposed? % WSL::
How do we determine if discovering a new answer is appropriate at all?
And how do we manage workers who lack the expertise or attention needed to provide useful input to a given task?
% WSL: This might be a dangerous line if not addressed quickly: a reviewer might quickly answer in their head "those dont take the task" -- later, we might need to deal with concerns that we're incentivizing people who don't know how to answer to take the task anyway. These issues can be addressed with real-time questions (where workers cant know everything that will come up in advance).
We present a mechanism that uses distinct payoffs for three possible worker actions---\emph{propose}, \emph{vote}, or \emph{abstain}---to provide workers with the necessary incentives to guarantee an effective (or even optimal) balance between searching for new answers, assessing  those currently available, and, when they have insufficient expertise or insight for the task at hand, abstaining.  We provide a novel game theoretic analysis for this mechanism and test it experimentally on an image-labeling problem and show that it allows a system to reliably control the balance between discovering new answers and converging to existing ones.
\end{abstract}

%\keywords{
%	Crowdsourcing; Game theory.
%}

%\category{H.1.2.}{User/Machine Systems}{Human information processing}

%%%%%%%%%%%%%%%%%%%%%%%%%%
% Intro
\section{Introduction}
\begin{figure}
\begin{center}
\hspace{-0.3pc}\includegraphics[width=20pc]{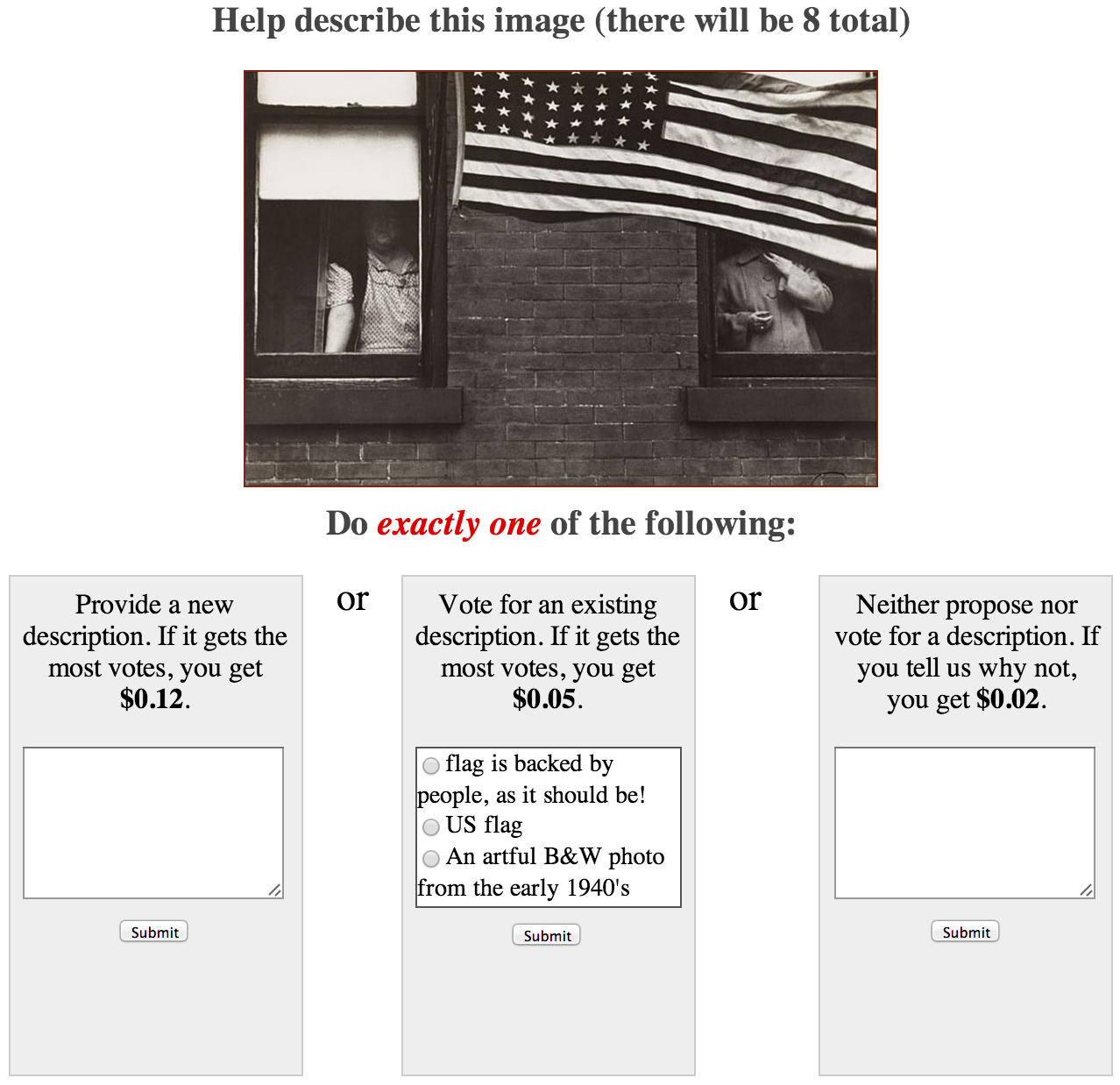}
\vspace{0.3pc}
\caption{An example of the sort of problems that the propose-vote-abstain can solve. Workers are asked to label a series of images. They are shown potential payoffs, and asked to choose between voting, proposing, or abstaining from contributing an answer.}
\vspace{0.8pc}
\label{fig:ui}
\end{center}
\end{figure} 

The rapidly increasing popularity of crowdsourcing is due in large part to its demonstrated capability to perform tasks that lie far beyond the reach of current fully automated systems, and to marshal human resources more efficiently than offline approaches can.
Prior work has shown that crowdsourcing can complete tasks such as editing documents \cite{soylent}, labeling images \cite{vonahndabbish2004}, discovering the folding structure of protein~\cite{foldit}, discovering celestial bodies \cite{lintott2008galaxy}, and evaluating the quality of products and services~\cite{alonsocrowdsearchengine2012}. Recent advances have even led to real time services \cite{legion,lasecki2013chorus,adrenaline,panosunknown2011,panoscontentandcontext2011}, such assistance technology that answers visual questions ~\cite{vizwiz,view} or help find accessible paths \cite{hara-curb} for blind users, and captions audio in any setting for deaf and hard of hearing users \cite{scribe}.

Across this wide range of applications, crowdsourcing relies on the mystery of human intelligence and the highly distributed nature of the crowd to gain its advantages. This leads to two basic design problems: how to motivate the crowd to contribute good information and how to integrate this information into a coherent solution.
%www.
Eliciting and integrating multiple workers responses has become a key attribute of many crowd-powered systems -- for instance, prior work has investigated how to coordinate and aggregate the input of multiple workers to determine the best edits to make in a document \cite{soylent},
combine real-time interface control input \cite{legion}, and determine the best response to provide in a conversation \cite{lasecki2013chorus}.
%... \cmh{help me out, Walter, with some examples.}

\begin{comment}
With the variety present in these different applications, finding a means of eliciting the best \emph{diversity} of responses is an important design problem. For example, crowds have been used as a source of on-demand feedback during the design process \cite{crowdCrit}, and for creating transcripts of from audio \cite{dual-path-transcription}.  While a user might want many different pieces of feedback from the crowd during a design critique, they most likely do not want multiple possible versions of a caption produced during a classroom lecture in a subject they are trying to learn.
\cmh{this isn't the best example for our research here: the example is about the amount of diversity present to the end-user, not to the crowd system as it aggregates input. The next paragraph tries to set up the problem as we see it here.}
\end{comment}

The relationship between data production and filtering is complex. In general, more workers contributing data means more filtering is necessary, which increases the overall cost of the task, in terms of (human or otherwise) computing resources. On the other hand, if too little effort is spent on discovery then the results may be unsatisfying. 
Additionally, if as here we assume that humans perform both tasks then we need to respect the competencies and preferences of these workers and route them to the particular tasks that they are most fit to perform. 

In this paper we explore the use of game-theoretic principals to achieve the right balance between data production and filtering. We present the \emph{propose-vote-decline} mechanism. 
Each crowd worker is given a choice among
\emph{proposing} an answer, \emph{voting} among the answers proposed
so far, or \emph{abstaining}, i.e., doing nothing. When a stopping
condition is reached, the mechanism returns the answer with the most
votes. Workers are paid a base amount, with bonuses if they propose or
vote for the winning answer.

% xxxxxxxxxxx|

We make the following contributions in this paper: 

\begin{itemize}
\item \textbf{(1)} We present the propose-vote-decline 
mechanism.  It is inspired by the natural interplay among individuals in any number of actual on- or off- line collective problem solving situations: individuals may suggest new ideas (or propose alternatives to existing ones), support exist ideas, or, if they feel they have nothing to contribute at the time, listen quietly. 

\item \textbf{(2)} We provide a game-theoretic analysis that shows, among 
other things, the baseline behavior of the system under different worker 
incentive structures. The analytical framework we use appears to be novel to the crowdsourcing world. Prior research focuses mostly on games of complete information that allow their workers to choose the cost (in terms of effort) they pay on their work. Our analysis uses fixed costs and accounts for worker uncertainty. We argue that these assumptions are better suited to many crowdsourcing settings, especially those involving \emph{microtasking} and real time services, where workers are not expected to put much effort into any one transaction.

\item \textbf{(3)} We study experimentally a fundamental problem 
in this setting, namely, how to best distribute limited resources between 
answer discovery and voting. How many alternatives---absolutely or relative to the number of voters---is ideal may vary dynamically on the answers proposed so far or on the nature of the questions asked. Figure \ref{fig:ui} shows to user interface of our experimental system. We show that the propose-vote-decline mechanism can provide information on what the best balance is for particular questions and
even tune this balance in real time.

\end{itemize}

%%%%%%%%%%%%%%%%%%%%%%%%%%
% Background
\section{Background}
\subsection{Crowdsourcing}
Crowdsourcing was coined by Howe \cite{howe2006rise}, though it existed long before then. He characterized it as any process having an open call format and a large network of potential laborers. It encompasses a wide range of activities, from collaborative filtering \cite{resnick1994grouplens} to system design contests \cite{moldovanu2001optimal}, citizen science \cite{foldit}, or any activity supported by Amazon Mechanical Turk (www.mturk.com) or Tasken (www.tasken.com). Workers may be engaged for minutes \cite{scribe} \cmh{Is this cite accurate?} or weeks at a time~\cite{foldit}. In this paper, we focus on \emph{microtask} crowdsourcing, which Amazon Mechanical Turk (AMT) and Tasken both support, in which workers are typically paid relatively small amounts of money (i.e., often less than a dollar) to perform relatively simple tasks with short time durations.

Microtasking has proved useful in writing and
editing \cite{soylent}, image description and interpretation
\cite{vizwiz,espgame}, and many
other areas. Existing mechanism focus on obtaining quality work
and generally introduce redundancy and layering into tasks so that
multiple workers contribute and verify results at each
stage. For instance, guaranteeing reliability through answer agreement
\cite{espgame} or the find-fix-verify pattern of Soylent \cite{soylent}.
%These abstractions introduce redundancy and layering into tasks to ensure reliability, but this takes time and is not conducive for real-time control.

Several systems have explored how to make crowdsourcing
interactive and real-time. As an example, VizWiz \cite{vizwiz} answers visual
questions for blind people quickly. It uses quikTurkit to prequeue
crowds of workers from AMT so that they
will be available when needed. For instance, the ESP Game
encouraged accurate image labels by pairing players together and
requiring them both to enter the same label, although ESP Game players
could also be paired with simulated players \cite{espgame}. Seaweed
reliably got AMT workers to be available at the same time
to play economic games by requiring the first worker to arrive to wait
(generally for a few seconds) \cite{seaweed}.

Recent crowd-powered systems target quick responses by pre-recruit workers who are then kept on standby until they are needed \cite{vizwiz,adrenaline}, which reduces or eliminates the time to recruit workers. Latency is thus determined by how quickly workers choose to complete their task. While current models sacrifice accuracy for latency, our approach is designed to encourage optimization of both accuracy and latency, based on a worker's internal confidence.

%CMH: This seems a bit self-serving and expansive, given the applicability of HiveMind to a much broader class of problems than continuous, real-time ones.
Another approach to soliciting real-time input is to maintain worker engagement in a task. This allows systems to benefit from repeated observation of the same workers and allows workers to exercise longer-term strategies. For instance, the Legion system connects crowd workers to existing desktop user interfaces via a real-time feedback loop for the duration of an interaction \cite{legion}. Input from multiple individual workers is aggregated into a single control stream by using worker responses to select the `best' answer in real-time. Legion has been used to reliably control a variety of interfaces requiring continuous real-time input, ranging from robot navigation to document editing. Similarly, Legion:Scribe uses multiple workers to perform a captioning tasks better than any constituent individual could have by synthesizing the workers partial captions into a single stream \cite{scribe}.

Systems such as Legion have generally motivate the crowd well, but requires tasks to be of fixed  length because workers are only rewarded only when the task is completed. In this paper, tasks are framed as atomic, with a single collective answer per task aggregated from one input from each worker. However, our approach supports chaining atomic segments together in a continuous-task setting, which enables application to continuous real-time tasks.
Segmenting of continuous tasks using the crowd  has been performed by
Legion:AR, using reliable groups of workers segmented and labeled a
video for activity recognition \cite{legion-ar}. Our approach provides
a means of motivating more general crowds for such tasks.

\subsection{Theory in Crowdsourcing}
Game theory provides an explanatory framework for reasoning about how the behaviors of others affects the outcome of interactive processes.
\emph{Mechanism design} is an area of game theory that provides an explanatory framework for designing \textit{games} or \textit{mechanisms} that elicit specific responses from players (participants in the situation---workers, in our case) thus yielding desirable outcomes. The players are assumed to be \textit{rational}, that is, players will always act to maximize their utility.

In spite of its formal conceits, game theory can help to explain group behavior in real settings, especially if we regard it as providing an idealized baseline against which to compare observed behavior. It has already been widely used to study a number of problems in crowdsourcing.

\begin{comment}
such as how to design trustworthy online mechanisms for ensuring trust \cite{escrow-mech, trust-online}. Such mechanism can be used, among other things, to filter out unreliable workers. Our mechanism also addresses this problem (though we do not study fully this aspect here) by providing financial incentives for workers lacking confidence in their expertise to drop out. Another game-theoretic approach to response-optimization is getting workers to increase their confidence levels by solving repeated tasks \cite{karger-budget}.
\end{comment}

Much of the game theoretic research on crowdsourcing falls into one of three categories. The first studies crowdsourced contests, in which participants compete in some sort of creative task, such as discovering the folding structure of protein \cite{foldit} (where workers compete for high scores) or design a new collaborative filtering algorithm for Netflix \cite{bell2007lessons} (where workers competed for a cash prize). This research typically views crowdsourcing as an all-pay auction \cite{dipalantino2009crowdsourcing}. Studies have compared the performance (in terms of market efficiency) of winner-take-all versus distributed reward structures \cite{boudreau2011incentives,moldovanu2001optimal,moldovanu2006contest,archak2009optimal,ChawlaHS12}. In most of this work, the amount of time in which a crowd worker is engaged can be very long, the value function for the contributions is public and objective, and which worker(s) are rewarded is determined by a principal actor. The major strategic problem the workers face is thus to determine how much effort they should expend. In this work, we are also interested in determining the best contribution(s).

In contrast, however, we are interested in situations in which evaluation is subjective in nature and there may be a great deal of uncertainty about what the ``best'' contribution is. We assume that the wisdom of the crowd is the best arbiter of quality---an assumption that is valid in many settings. We also assume that workers are engaged for short periods of time and paid a small amount for their contributions. These assumptions hold true in many crowdsourcing applications. Yet under these assumptions it is not realistic for workers to choose strategically the amount of effort they will expend on their contributions, as the range of effort levels available to choose from is quite narrow. Rather, workers can choose strategically whether to propose a new contribution, endorse (by voting for) an existing contribution, or abstain altogether, and our goal is to design mechanisms that provide the best answer available on demand and quickly converge to an optimal answer.

The second major category of game theory applied to crowdsourcing deals with recommender systems. Like the present paper, much this work  \cite{resnick1997recommender,prelec2004bayesian,miller2005eliciting,jurca2005enforcing,jurca2007robust,witkowski2012peer}, assumes that rating is a subjective process and use voting to perform evaluation. Unlike our work, much of this research assumes that the alternatives to evaluate are predetermined. A notable exception is Gao et al.~\cite{gao2012quality} who, like us, consider jointly the proposal and voting processes. However, unlike us, they consider both processes as distinct with distinct workers assigned to distinct roles as voters or authors, and where the voters evaluate the authors, not their contributions. By contrast, we allow the workers to determine the roles they wish to play and rather consider the effect different incentives have on the roles they chose.

Nearly all of the previously-mentioned work assumes that the workers are synchronized and that they play their moves simultaneously. By contrast, we assume that workers are asynchronous and play sequentially. This makes the analysis somewhat more complex, however it allows to study real-time concerns.

The third major stream of game-theoretic research considered control-theoretic mechanisms, primarily in settings where workers are engaging over relatively long periods of time, as in citizen science. Kamar et al. use such partially-observeable Markov decision processes to study a crowdsourcing problem in which the workers procede in a series of rounds, and describe an efficient algorithm that, in their case, learns an optimal policy~\cite{crowdsynth}.

Social choice theory has also been used in the context of crowdsourcing, say, to power collaborative filtering.
Parkes and Procaccia study social choice functions~\cite{parkes2013dynamic} embedded in Markov decision processes. These systems consider game-theoretic behavior from the requester's rather than the worker's perspective. We, on the other hand, are concerned with the strategies of the workers.

Hemaspaandra et al.~\cite{hemaspaandra2013complexity} consider online voting where the agents vote sequentially and each agent may know the votes of the preceding agents, but not the succeeding ones. The alternatives to choose from, however, are know in advances to all agents. They are primary concerned with manipulation in the setting. In our setting, the alternatives available change dynamically over time. Also, we do not directly study manipulation in this paper. We discuss both of this issues in the limitations section.

Like most social choice research, both works above assume that all alternatives are known to all the beforehand. By contrast, the dynamic discovery of the alternatives (proposals) is a central goal of our system, and so we cannot assume that the alternatives are known beforehand.

%%%%%%%%%%%%%%%%%%%%%%%%
% The Mechanism
\section{The Propose-Vote-Abstain Mechanism}

Each propose-vote-abstain \emph{round} begins with a \emph{request}---for information or an answer to a
specific question---by the \emph{requester}. Workers are then recruited into the system.

As soon as each worker joins, he or she is presented with the request and all the contributions proposed so far by the previous workers. The new worker is then offered  the following choices: \emph{propose} a new contribution, for a payoff of
$\pi$ if that answer eventually receives the most votes,  \emph{vote}
for one of the contribution that has already been proposed, for a payoff of $\nu$ if that
contribution eventually receives the most votes, or abstain for an unconditional payoff of $\alpha$. If at termination more than one contribution has the most votes, then one of them is selected uniformly at random as the winner.

Crucially, no worker knows the number of votes each existing contribution has received, nor how many workers have already played, nor how many workers total have been recruited. This helps curb the \emph{majority effect}, though obviously a strategic player would vote for the candidate believed to be the one most preferred by the other voters.

The process continues until some stopping condition is reached. This could be after a fixed number of workers respond, or one (or more) contribution receives a certain minimum number of votes.

In this paper, we assume that the payoff structure $(\pi, \nu, \alpha)$ is fixed per round (i.e., the same for all workers) and only consider single round settings, i.e., each worker makes exactly one move (vote, propose, abstain). Multiround settings are also natural to consider, especially in the context of highly interactive, deliberative group problem solving, or in realtime streaming situations when the mechanism must continuously emit solutions (or receive new requests based on previous solutions). However, single round settings are common enough to use them in this first study on the mechanism. 

It is also reasonable to consider variations in which $(\pi, \nu, \alpha)$ vary dynamically within a single round. Likewise, the requirement that only one alternative be chosen or that plurality is the social choice function can be generalized. Such considerations are beyond the scope of this paper.

%%%%%%%%%%%%%%%%%%%%%%%%%
% Game Theory
\section{Modeling Worker Behavior}
From the theoretical perspective we view the propose-vote-abstain mechanism as a single round, sequential-play, extensive-form game.
Doing so helps us to reason about how different incentive structures may
lead to different outcomes. Here we provide a number of baseline results,
and give a few glimpses into how the same framework can be used to model
real-world behavior.  Our assumptions are inspired by a scenario that is arguably not likely to occur in practice, but which here provides very useful baseline.

\subsection*{Catastrophic Freeloading}

We call this scenario \emph{Catastrophic Freeloading}. It describes a situation in which no worker is willing to put any effort into proposing or voting, but is willing to randomly vote or propose if doing so yields a higher expected payoff than abstaining. We describe it in terms of a belief model in which each worker is completely oblivious to the choices available and, furthermore, believe that all other workers are similarly obviously. 

This scenario is important to consider because it describes the amount of data produced by a completely ignorant and unmotivated (but strategically clever) workforce and estimates the cost (in worker payments) to the system for that data.

In particular, we make the following assumptions.

\begin{enumerate}
\item \label{it:indy} The game has a single turn with a indeterminate (and unknown to
  the workers) number of players.
  \item \label{it:sequence} Workers play sequentially, one at a time.
\item The game does not terminate unless there is at least one answer
  proposed and one vote. 
\item \label{it:uni} We assume that all players have equal confidence in any
  alternative being better than another, including those the players themselves may propose. 
\item \label{it:info} The only information the workers know about the system are the
  candidate answers and the intial request for proposals.
\item \label{it:fixed-cost} The tasks required of the worker are simple and require a fixed time and energy commitment.
\end{enumerate}

Assumptions (\ref{it:uni})--(\ref{it:info}) mean effectively that the only information we have to distinguish one state (i.e., set of candidate answers) from another is the number of alternatives $m_t$ 
proposed so far, where $t$ is the current time. (\ref{it:uni}) lets us assume that the votes cast at any time are uniformly distributed at random from among all alternatives available at that time. That is, \emph{we assume here that voting is effectively random selection}.  Essentially, this means that all workers are completely ignorant about the alternatives available to them, and they believe that all the other workers share this level of ignorance. Thus, the most effective way to choose an alternative is by blind guessing. It ignores shared biases, basic and expert knowledge, irrational behavior, and many other factors that might lead to more complex behavior. 

These catastrophic freeloading assumptions are obviously very strong, so much so that any notion of ``preference'' over the set of proposals is all but erased. We take it to be a baseline condition, against which assumptions accounting for worker preferences can be measured. This is, admittedly, a small step, but a necessary one, for how to model worker preferences and---as is crucial in game theoretic considerations---\emph{beliefs about other workers' preferences}---is not a straightforward problem, nor one that can be easily generalized across a potentially broad spectrum of potential applications.

An alternative view of these assumptions is they reflect the actions of a worker that is completely disengaged in the process, i.e., ``phoning it in.'' It could be seen as an extreme case where all workers are freeloaders who want to maximize their expected reward while spending a minimum amount of time.

Regarding (\ref{it:fixed-cost}), in prior, auction-based research \cite{moldovanu2001optimal} workers had complete information about the game and were allowed to choose strategically the amount of effort to expend on the problem. Under those assumptions the purpose of the mechanism was to provide inventives for the voters to provide the right amount of effort. In microtask settings, such as the one on which our experiment is based, workers tend to be engaged for short amounts of time and are not expected to bring much effort to any one transaction. Thus, from a design perspective it is much more important to focus on the \emph{beliefs} and \emph{knowledge} the workers can bring to the task right now, as opposed to the amount of effort they could spend on more open-ended tasks.

With these assumptions in mind, Theorem \ref{thm:main} tells us quite a bit about the baseline behavior of the crowd, from a mechanism design perspective. In particular, it considers, given a 

\begin{theorem}\label{thm:main}
Given some target number of candidate proposals $m \in \mathbb{N}$, any cost structure $(\pi, \nu, \alpha)$ that satisfies $(m+1)\cdot \alpha > \nu > m \cdot  \alpha$ and $(m+1) \cdot \nu > \pi > m \cdot \nu$ generates a dominant strategy in which the first $m$ workers propose and the remaining workers vote.
\end{theorem}

Theorem \ref{thm:main} follows from the propositions below. It tells us that, when workers are maximally uncertain about the choices at hand then (1) the abstention payoff provides an incentive for limiting the number of proposals, \emph{though not the number of votes}; and (2) the workers monotonically propose first and vote later, never returning again to propose.

Before we present these propositions, let us first provide some useful notation and observations. Let $m_t$ denote the number of candidate proposals present for the $t$th worker to choose from. Note that under the above assumptions this is effectively all that the worker observes about the candidate proposals, as he or she is completely oblivious to the content of each proposal. Thus, all we need to know about the state in order to model worker behavior is $m_t$ for each worker $t$. Thus, for any other worker $t'$, if $m_t = m_t'$ then $t$ and $t'$ have the same dominant strategy. In particular, if some dominant strategy results in $m_{t+1} = m_t$ then all remaining workers will repeat this action for the rest of the game. This is the case whenever the dominant strategy at $m_t$ is to vote or to abstain. If the dominant strategy is to propose, then $m_{t+1} = m_t + 1$. This simple observation is key to understanding the propositions. 

Let $S^*(m_t)$ denote the dominant strategy for $m_t$. 

Proposition \ref{prop:basics} provides basic results about worker strategies, in particular sheds light on the role that abstention plays.
\begin{proposition}\label{prop:basics}
Under assumptions (\ref{it:uni})--(\ref{it:info}):
\begin{enumerate}
\item \label{it:abstain}  If $\alpha \geq \min\{\pi, \nu\}$ then abstaining is a dominant strategy for all players.
\item \label{it:abstain_not_propose}  If $\min\{\pi,\nu\} / (m_t + 1) < \alpha$ then for worker $t$ abstaining has a greater expected payoff than proposing.
\item \label{it:abstain_not_vote} If $\nu/ m_t  < \alpha$ then for worker $t$ abstaining has a greater expected payoff than voting.
\begin{comment}
\item \label{it:vote_is_dominant} If after proposing, the expected payoff for anyone voting afterwards would be less than the payoff for abstaining, then the expected payoff for proposing is $0$.
\end{comment}
\end{enumerate}
\end{proposition}
\begin{proof}
(\ref{it:abstain}) is obvious.

(\ref{it:abstain_not_propose})---(\ref{it:abstain_not_vote}) follow from assumption (\ref{it:uni}), in that having uniform confidence in each candidate--and assuming that all workers feel similarly---means that any voter is  equally likely to chose any candidate currently available, so that the expected payoff for proposing or voting is upper bounded by dividing the base reward $\pi$ or $\nu$  among the number of proposals available to vote on after the worker moves, which is $m_t + 1$ if worker $t$ proposes and $m_t$ if worker $t$ votes.  For (\ref{it:abstain_not_propose}), note additionally that  whenever proposing would create so many candidates as to give voting a lower expected payoff than abstaining, proposing has a lower expected payoff than abstaining.
\end{proof}

Next, we study conditions that lead to proposal and voting.

\begin{proposition}\label{prop:one}
If $\nu > \pi > \alpha$ then $S^*(1) = \mathrm{vote}$.
\end{proposition}
Proposition \ref{prop:one} is obvious.

While the earlier propositions provide important boundary results, Proposition \ref{prop:dynamics} provides insight on the system state as the game unfolds. The proposition works backwards, but finding a state in which workers vote. Since voting does not change the state, the system will remain in this state until termination. We then use structural induction to work backwards from this voting state to characterize the rest of the state space.

\begin{proposition}\label{prop:dynamics}
If $\pi > \nu > \alpha$ then:
\begin{enumerate}
\item \label{it:whenabstain} If $m_t > \nu/\alpha$, then $S^*(m_t) = \mathrm{abstain}$.
\item \label{it:whenvote} If $m_t + 1 > \nu/\alpha > m_t > 0$ then $S^*(m_t) = \mathrm{vote}$.
\item \label{it:when_prop_or_vote} For all $t,w, m_t, m_w \in \mathbb{N}$ such that $w > t$, $m_w > m_t > 0$, $S^*(m_w) = \mathrm{vote}$ and, for all $u : m_w > m_u > m_t$, it holds that $S^*(m_u) = \mathrm{propose}$: 
\begin{enumerate}
 \item If $m_t > \nu \cdot m_w /\pi$ then $S^*(m_t)  = \mathrm{propose}$.
 \item If $m_t < \nu\cdot  m_w /\pi$ then $S^*(m_t)  = \mathrm{vote}$.
\end{enumerate}
\end{enumerate}

\vspace{.3cm}

Let $m^{\nu}_0$ denote $\lfloor \nu/\alpha \rfloor$ and, for all $i \geq 0$, let $m^{\nu}_{i+1} = \lfloor \pi \cdot  m^{\nu}_i / \nu \rfloor$.

\begin{enumerate}
 \setcounter{enumi}{3}
\item \label{it:whenever-vote} For all $i \geq 0$ if $m^{\nu}_i > 0$ then $S^*(m^{\nu}_i) = \textrm{vote}$. 
\item \label{it:whenever-propose} For all $m_t$ such that $m_t \leq \nu/\alpha$ and $m_t \not\in \{m^{\nu}_i\}$, $S^*(m_t) = \textrm{propose}$.
\end{enumerate}
\end{proposition}

\begin{proof}
Item (\ref{it:whenabstain}) follows from Proposition \ref{prop:basics}.\ref{it:abstain_not_vote}.

For item (\ref{it:whenvote}), if $m_t + 1 > \nu/\alpha > m_t$ then voting has a higher expected payoff than abstaining, as long as no one else proposes. But proposing would make the expected payoff for voting less than for abstaining. Consequently, the expected payoff for proposing is zero at this point (assuming all subsequent workers act to optimize their expected payoffs). Thus voting has the highest expected payoff.

For item (\ref{it:when_prop_or_vote}), if $m_t > \nu \cdot m_w / \pi$ then $\pi/m_w$, the expected payoff for proposing is greater than  $\nu/m_t$ the expected payoff for voting, which is greater than $\alpha$ (otherwise $S^*(m_w) \neq \mathrm{vote}$). If $m_t < \nu\cdot  m_w /\pi$ then by a similar argument the expected payoff for voting is greater than proposing.

Items (\ref{it:whenever-vote}) and (\ref{it:whenever-propose}) follows by applying item (\ref{it:when_prop_or_vote}) inductively to each $m^{\nu}_i$.

\end{proof}

Items (\ref{it:whenever-vote}) and (\ref{it:whenever-propose}) show that even in the relatively small state space we consider here that the space of dominant strategies is fairly rich, varying between proposing and voting for different values of $m_t$. However, little of it is reachable in practice, as the state does not change one workers start voter.

\begin{comment}
\cmh{The following paragraph needs a lot of work.}
These theoretical results help us to hypothesize about how workers might behave under more realistic circumstances. For instance, it gives us an upper bound on the number of proposals we can expect to get even if users are completely ignorant about the situation. We also see that the only reason a worker would abstain is (1) if the worker has little confidence in either his or her own potential contribution or those that have been so far proposed (note that even this lack of confidence shows a stronger degree of belief about the proposals than complete ignorance would show) or (2) a previous worker had proposed a contribution that the previous worker felt very confident about (so much so that the number of proposals is presently so large that random voting will no longer yield a greater expected outcome than abstaining). The results suggest that, the higher the payoff abstaining is (though the payoff for abstaining must be lower than the other payoffs for any proposing or voting to happen) the less random voting will appear. 
\end{comment}

Finally Proposition \ref{prop:initial} provides initial conditions.

\begin{proposition}\label{prop:initial}
If $\alpha < \min\{\pi, \nu\}$ then $S^*(0) = \textrm{propose}$.
\end{proposition}
\begin{proof}
If worker $0$ proposes then its expected payoff is greater than abstaining or voting as long as there are never more than $\min\{\nu, \pi\}/\alpha$ proposals. But this will never happen because of Proposition \ref{prop:basics}.\ref{it:abstain_not_propose}.
\end{proof}

With these results in mind, we can now prove Theorem \ref{thm:main}.

\begin{proof}[Proof (of Theorem \ref{thm:main})] For a given target number of proposals $m$ and cost structure $(\pi, \nu, \alpha)$ that satisfies $(m+1)\cdot \alpha > \nu > m \cdot  \alpha$ and $(m+1) \cdot \nu > \pi > m \cdot \nu$, note that the only $i \in \mathbb{N}$ for which $m^{\nu}_i > 0$ is $i = 0$. Thus, by Propositions \ref{prop:dynamics}.\ref{it:whenever-vote}--ref{it:whenever-propose}, $S^*(m) = \mathrm{vote}$ and for all $m'$ between $0$ and $m$ (exclusive), $S^*(m') = \mathrm{propose}$. By Proposition \ref{prop:initial}, $S^*(0) = \mathrm{propose}$.
\end{proof}

What conclusions should a system designer draw from these results? Certainly, in order to making voting or proposing incentive compatible actions, it is obvious that voting must pay more than abstaining, and proposing must pay more than voting. It is also reasonable to assume that the more one pays for each action, the more effort workers will spend to perform them. These results show that there is a another side to the matter: that raising incentives for proposing and voting can also increase the amount of bad data gathered. Of course, the catastrophic freeloader scenario we study here may not apply often in practice. Even when there are freeloaders, they may not dominate the system, and even when they do, it is not clear that the freeloaders would themselves be aware of their dominant situation and so might not adopt the strategies seen here (that is, they might behave differently if they believed that some of the other workers are putting effort and attention into their work). 

The results also suggest that conscientious, informed workers will act when freeloaders may not, and so, to avoid freeloading, the best approach is to make voting pay more than abstention, but less than twice as much, and make proposing pay twice as much as voting, but less than twice as much. Under such a cost structure, Theorem \ref{thm:main} says that under the catastrophic freeloading modeling, one freeloading might propose (and any number of them could vote as long as there is only one proposal), but any remaining proposals and votes would come from workers who more sincerely believe in their choices.
%%%%%%%%%%%%%%%%%%%%%%%%
% Experiments
\section{Experiments}

\begin{table}[h]
\centering

\begin{tabular}{r|r|rrrr}
 $\pi$ & $\nu$ & workers & proposals & votes  & abstains\\
 \hline
 \$0.20&\$0.04& 20 & 39 & 60 & 1\\
 \$0.12&\$0.05& 33 & 41 & 117 & 7\\
 \$0.08&\$0.08& 20 & 13 & 86 & 1\\
 \$0.05&\$0.12& 10 & 13 & 36 & 0\\
 \$0.04&\$0.20& 18 & 13 & 76 & 3\\
\end{tabular}
\caption{Basic statistics from the experiments. The payment for abstaining is fixed at \$0.02. Each HIT asks each worker about the same five images, shown in Figure~\ref{fig:images}, for proposal and voting payoffs $\pi$ and $\nu$, respectively.}
\label{tab:basic}
\end{table}

We recruited 111 Mechanical Turk workers and asked them to view a set of 5 images (presented in random order) and either propose, vote for, or abstain from contributing to the image's description.  A round consists of a single image with a fixed (i.e., the same for all workers in the round) payoff structure (i.e., fixed prices) for proposing, voting, and abstaining. All five rounds assigned to each worker had the same price structure There were in total five different payoff structures, leading to 25 rounds total.

Table~\ref{tab:basic} shows the five payoff structures and how the workers and their actions distribute over them.

\begin{figure}
\centering

\begin{tabular}{c}
\begin{tabular}{c}
IMG1\\\includegraphics[width=10pc]{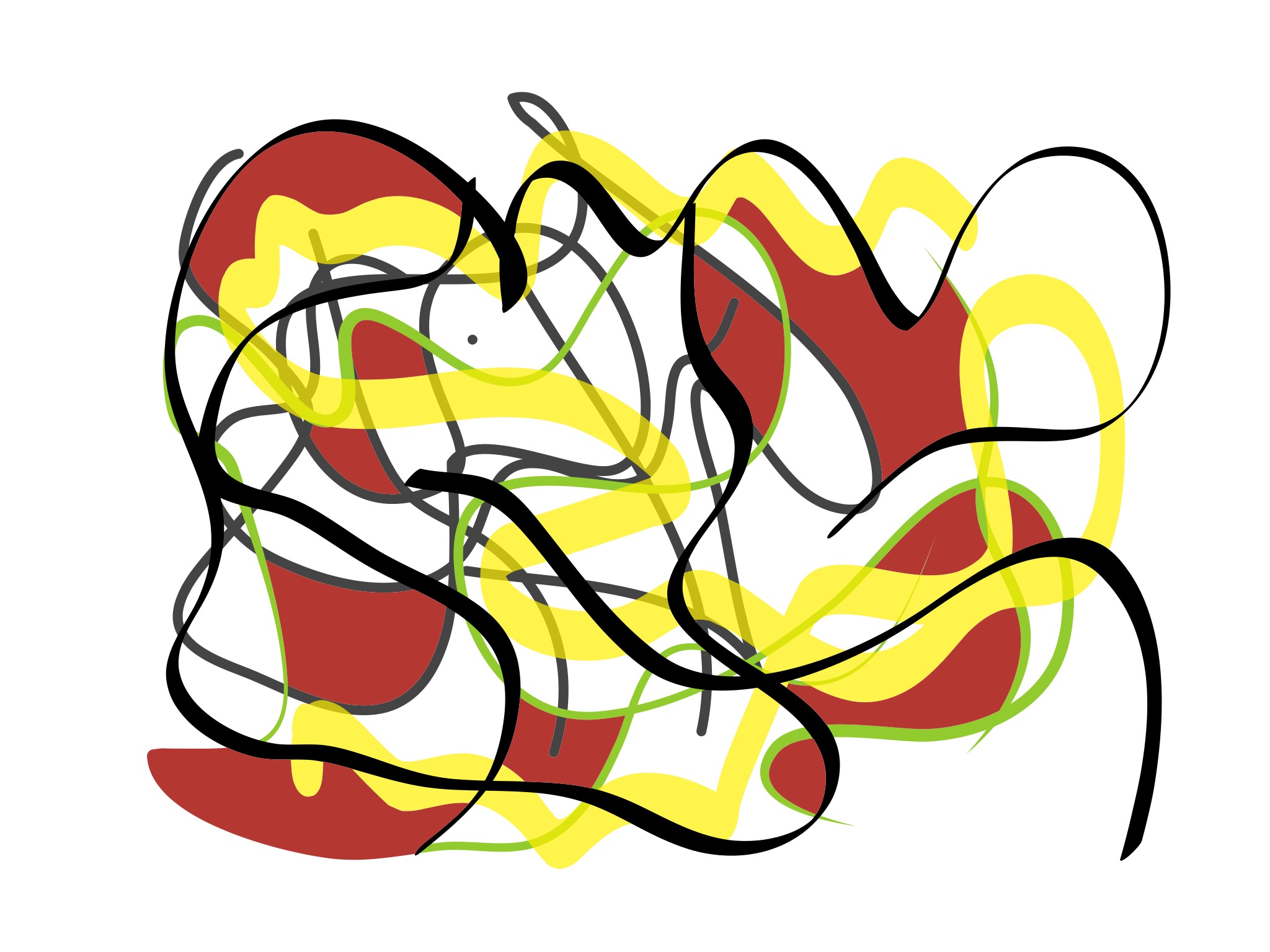} \\
IMG2\\\includegraphics[width=10pc]{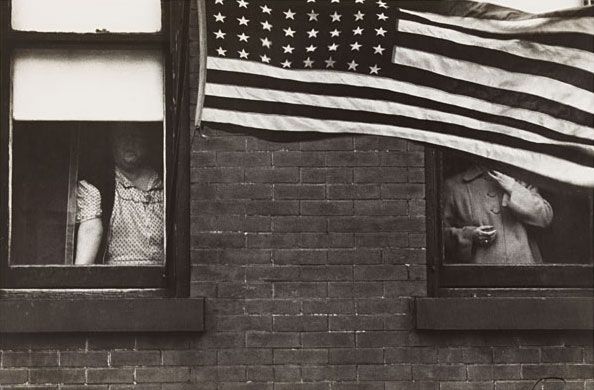} \\
IMG3\\\includegraphics[width=10pc]{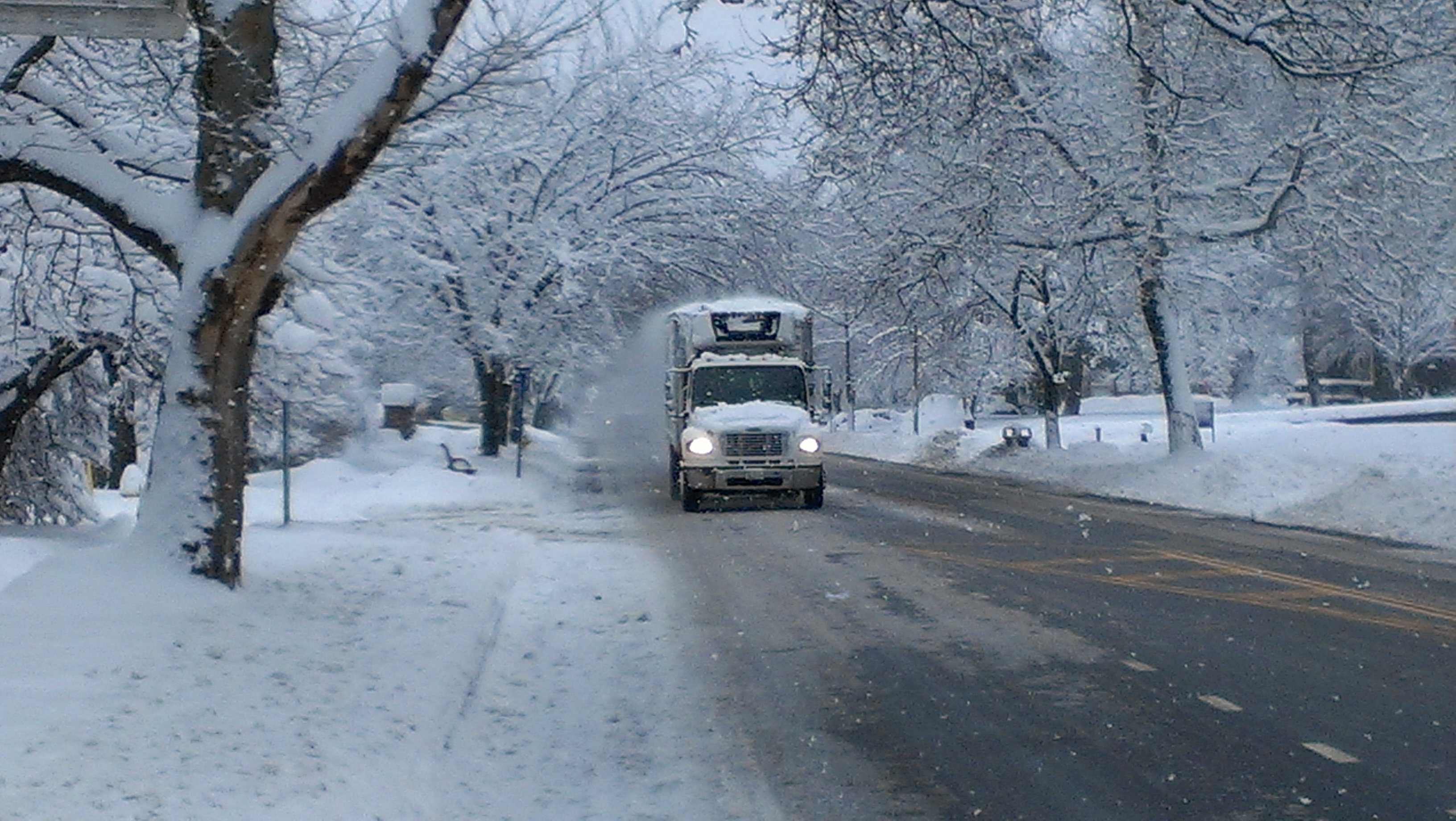} 
\end{tabular}\\
\begin{tabular}{cc}
IMG4 & IMG5\\
\includegraphics[width=6pc]{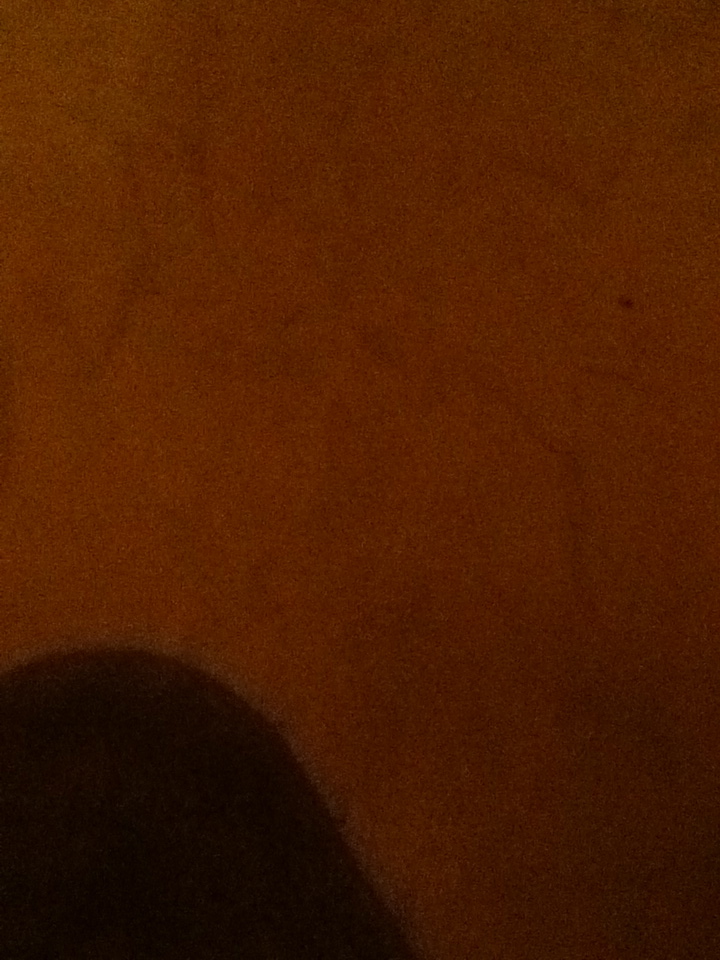} &
\includegraphics[width=6pc]{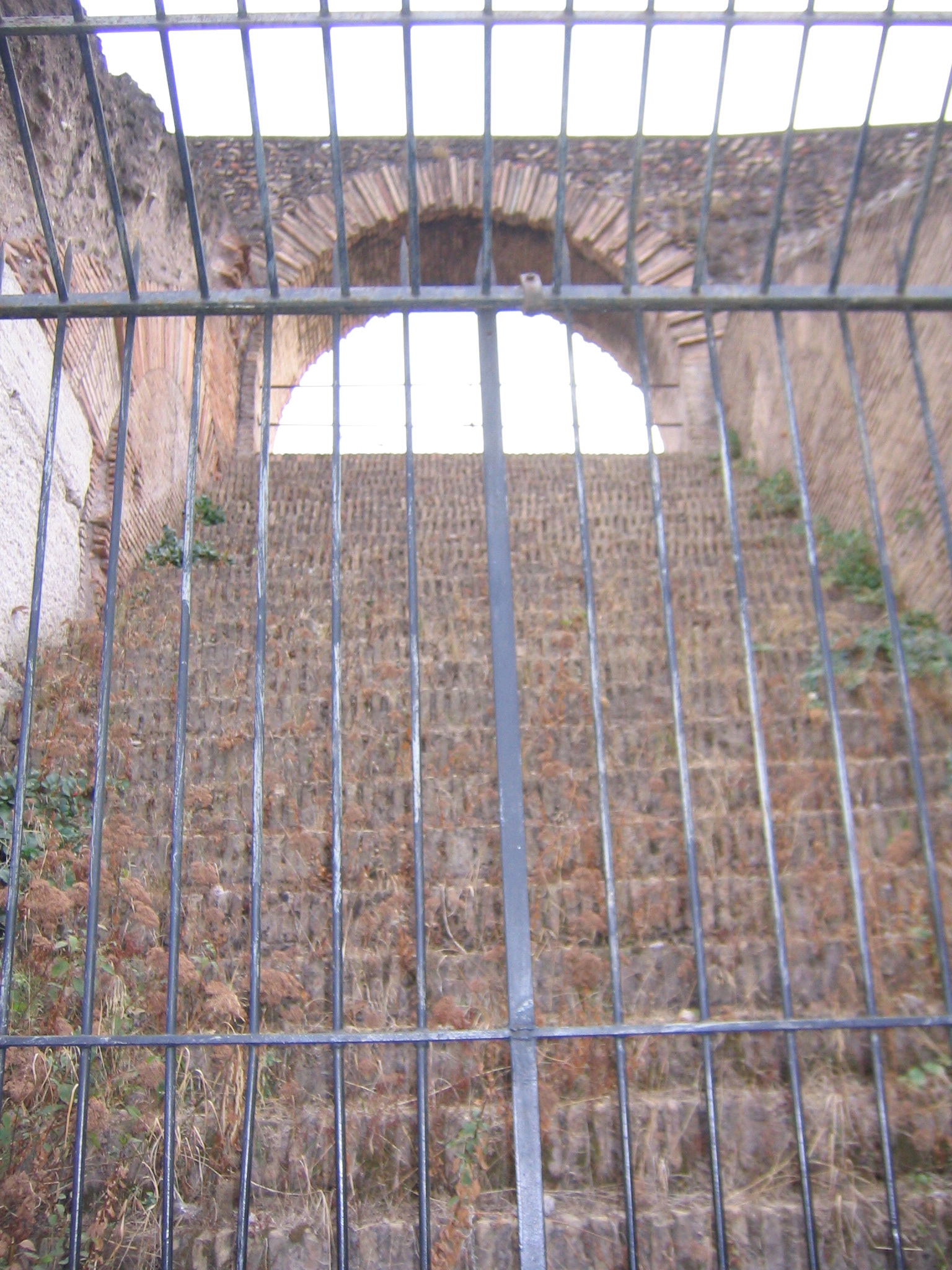} 
\end{tabular}
\end{tabular}

\caption{The images used in the experiments.}
\label{fig:images}
\end{figure}

Figure \ref{fig:images} shows the images to be labeled. In order to not bias the experiments towards specific domains, the images were chosen to be vague.

\begin{figure}
\begin{center}
\hspace{-0.3pc}\includegraphics[width=20pc]{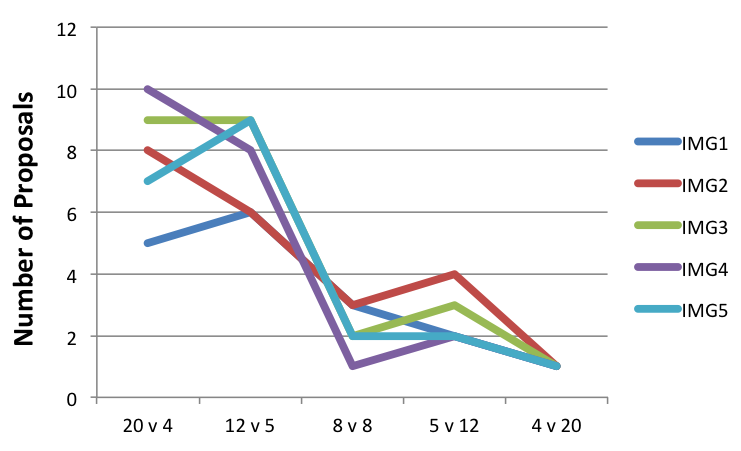}
\vspace{0.3pc}
\caption{Results from varying relative payments for voting and proposing responses to an image description task.}
\vspace{0.8pc}
\label{fig:pics_res}
\end{center}
\end{figure}

Figure \ref{fig:pics_res} shows the results of our tests on the number of proposals generated. As the payment for voting becomes large relative to the proposal payment, the number of total answers generated by the system significantly decreases from an average of 7.8 responses to 1 response for all 5 pictures we saw ($p < .0001$). The decreasing trend was linear with $R^2$ = 0.802. Note that there is a disproportional drop at the break-even point when payment is equal for both options. This is consistent with what we expect because voting requires less effort than generating a response, so there is a slight bias in its favor.

While all of the images eventually converged to a single response as the vote payment increased, the number of responses generated in the opposing case (where the proposal reward is high and workers are incentivized to generate several answers) varies from 5 to 10 responses each. This likely depends on how subjective the image is and how many answer could be considered plausible with high confidence. This trend is seen throughout the results as each response set trends towards a single response. This suggests that the content of the task does play a role in worker trends, but in the convergent limit this can be overridden by financial incentives.

\subsection{Overvoting}

\begin{table}[h]
\centering
\begin{tabular}{rr|rr}
& & number of & number of \\
&& voting & actual \\
 $\pi$ & $\nu$ & workers & votes\\
 \hline
 \$0.20&\$0.04& 60 & 69 \\
 
 \$0.12&\$0.05& 117 & 147 \\
 
 \$0.08&\$0.08& 86 & 90 \\
 
 \$0.05&\$0.12& 36 & 36 \\
 
 \$0.04&\$0.20& 76 & 124\\
 \hline
 &total& 375 & 464
 
\end{tabular}
\caption{The number of votes registered for each round in each payoff structure compared to the number of workers who choose to vote. Overvoting occurs when workers reload a screen and can be accidental or intentional. Lines two five seem to show evidence of cheating. The payment for abstaining is fixed at \$0.02.}
\label{tab:my_label}
\end{table}

One problem we did not expect to encounter to such a degree was overvoting (mainly because we assumed that workers would expect for us to check for this). It is possible for workers who choose to vote to do so more than once to do so simply by hitting the ``back'' button on their browser. It is not clear whether this happens accidentally or intentionally, however, Table~\ref{tab:my_label} shows the total number of over- versus actual- votes in the rounds of each price structure. The amount of overvoting in the second and fifth price structures in particular seems too high to be solely accidental.

\subsection{Theoretical Predictions}
\begin{table}[h]
\centering
\begin{tabular}{rr|rrr}

&&\multicolumn{3}{c}{\# rounds where \# proposals exceeds}\\
$\pi$ & $\nu$ & $\pi/\alpha$ & $\min\{\pi/\nu, 1\}$ & $\nu/\alpha$ \\
\hline
20& 4 & 0 & 5 & 0 \\

12& 5 & 2 & 5 & 5 \\

8& 8 & 0 & 0 & 4 \\

5& 12 & 0 & 0 & 2 \\

4& 20 & 1 & 0 & 4 \\
\hline
&\textbf{total} & 3 & 10 & 15 \\
\hline
&winners & 0 & 4 & 4

\end{tabular}
\caption{How does actual worker behavior compare to the theoretical predictions under the zero knowledge assumption? This table looks at three values that upper bound the number of proposals under those predictions. The values are based on the payoffs for proposing, $\pi$, voting, $\nu$, and abstaining $\alpha$. Note that the maximum number of rounds for each price structure is five.}
\label{tab:pred}
\end{table}
How did the workers in our experiment fare against the theoretical predictions? The theory predicts that whenever the payoffs for proposing and voting are greater than the payoff for abstaining then the first worker will propose and at least one worker will vote. This was consistent with our experiments. 

It also predicts that when the payoff for abstaining is greater than either voting or proposing then no one will vote or propose. We used such a payoff structure in one of our training rounds, and in none of those cases did a worker propose or vote.

Additionally, the theory predicts that there would be no abstentions. Table \ref{tab:basic} shows that, in fact, few workers ever abstained.

Moving beyond these basic results, the theory predicts that, under the catastrophic freeloader, there are two different values that upper bound the number of proposals: namely
$\pi/\alpha$ and $\nu/\alpha$, and $\min\{\pi/\nu, 1\}$, where $\pi$, $\nu$, and $\alpha$ are the payoffs for, respectively, proposing, voting and abstaining. Table \ref{tab:pred} shows the number of rounds where the number of proposals exceeds at least one of these bounds. 

These results, though admittedly of a preliminary nature, might support the hypothesis that $\pi/\alpha$ is an effective upper bound on the proposals.  In terms of our worker's beliefs (and under standard bounded rationality assumptions), this could suggest that, when it comes to proposing new contributions, workers have little confidence in their proposals.

The results also show that the remaining bounds---$\nu/\alpha$, and $\min\{\pi/\nu, 1\}$---having do with how with the expected payoff to future voters and thus about what the current worker thinks of the competence of future workers to vote for his or her contribution, if the current worker proposes, are frequently violated. This would seem to suggest that workers, even though they are not confident about their ideas, are rather more confident about the ability of future voters to make the right choice.

Another takeaway is that the number of times a winning contribution proposed after any one of the three boundaries is violated is very small (eight total). The theory suggests that a rational worker would violate these boundaries only if the worker was confident in having an answer. Our results might suggest that such a worker's confidence is often misplaced.

\begin{figure}
\centering
\includegraphics[width=20pc]{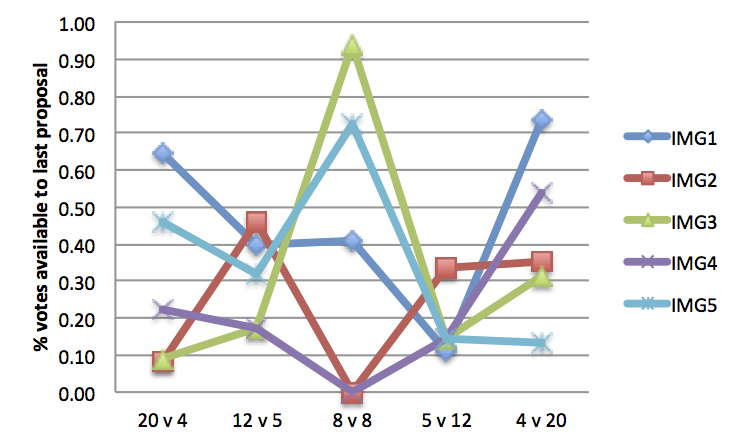}
\caption{Percentage of votes available to the last contribution. In our experiments we observed that---contrary to prediction---workers do not all propose first and then vote after all proposals appear; proposals and votes are interleaved. Thus, proposals that appear earlier are generally eligible for more votes than those that appear later.}
\label{fig:lastprop}
\end{figure}

Regarding monoticity, we looked at the order in which votes and proposals appear. Rather than all proposals appear first, followed by all votes, with no interleaving, votes and proposals were in face interspersed. Consequently, later proposals were eligible for fewer votes. Figure \ref{fig:lastprop} show the percentage of voting workers who did so after the last contribution was proposed. 

%%%%%%%%%%%%%%%%%%%%%%
% Discussion
\section{Discussion}
\begin{comment}
One key way in which the experimental performance of the propose-vote-abstain model deviated from our theoretical predictions is that workers proposed more often in the experimental setting. This suggests that workers have more confidence in their answers than they would under the zero-knowledge assumption, and this is what we would expect. But is that confidence well placed? If so, then we would expect to see only a few proposals beyond what the theoretical limits would predict, and would expect those proposals to win. Instead, we often saw in a given auction a number of proposals beyond the expected limit, and they often did not win. This may suggest the false-consensus effect, or the misplaced belief that one's views are widely shared with others \cite{bauman2002we}.
\end{comment}

Among our incoming assumptions was that the time commitment and financial benefits to microtasking are so small that workers do not make decisions about the amount of effort to put into a task. The possible existence of cheating among voters suggests that worker effort, even at this low-stakes level, must still be considered. We are curious about the motivations of the cheaters. Do they think that cheating gives them a leg up on their competitors, or do they assume that everyone else is doing and so they need to follow suit just to keep up with their peers?

From a theoretical perspective, voting systems in which the alternatives to choose from and votes themselves are added asynchronously, in sequence, presents a number of interesting challenges. For instance, in general there is a bias towards the alternatives that appear earlier, since they have a long time frame within which to gather votes. There are several ways to mitigate this bias, e.g., by fixing the interval of time between when a contribution is proposed and when it can last be voted on, so that all contributions have the same amount of time to garner votes. One can also use importance sampling or impose halting conditions that guarantee each proposal gets a fair chance of being selected. How to best do so depends on the application at hand and is a matter of further study.

\subsection{Limitations}
There are a number of limitations to our approach. The first is that game theory generally imposes a number of simplifying assumptions on the behavior of individuals. This is usually necessary to make the theory tractable, but often leads to behavior that is not observed in practice. Our view hear is that it can nonetheless go a long way in explaining observed behavior if it is regarded as presenting an ideal case (or null hypothesis) and we reflect on \emph{how} (or if) observed behavior deviates from it.

Another limitation is that we do not consider the possibility that workers could collude, say, to propose an answer and then all vote for it. Clearly, the incentive mechanism does not account for this. Voting is well known to be vulnerable to manipulation under ideal conditions, however, how vulnerable they are in practice is a matter of ongoing research. Another way around this problem is to start with a baseline assumption about how large such a manipulating coalition can be and then make the worker pool large enough to negate any effect the manipulators might have. Finally, one can user worker feedback systems to punish those who manipulate through collusion.

Another limitation is that we are taking the ``best'' answer to be that chosen by the majority of the workers. Though this can be effective in certain circumstances---such as putting general labels on images---but not as well when expert advice is needed. For instance, if a worker happened to know know that IMG2 was a by the artist Robert Frank, would he or she propose this information? Or would he or she hold back and propose something simpler (or not propose at all), expecting this fact to be lost on the average worker? 

While we might expect the latter behavior to dominate, we at least once observed the former behavior. In one round involving IMG5, the proposal that received the most votes was a later one that read, ``Metal gate blocking a stone staircase 
that leads up to an ancient aqueduct.'' Perhaps the specificity of the answer conveyed enough authority to convince voters to support it. (Incidently, the image is of the Colosseum of Rome, so the proposal is not too far off the mark.) 

The rest of the proposals from this round are illuminating to read, as they convey a variety of strategies in proposing answers. They read, in order of appearance, ``gate'', ``Protect the place'', ``Staircase'', ``Fenced steps'', (the winning proposal), ``steps to prison'', ``This is a metal grate installed over the staircase to one of Rome's deteriorating  ruins'' (this proposal is slightly less specific than the winning one, and also less erroneous), ``Ruins are ruins, why want to peep in? 
Hence the grill-grate!!'' (which has a less authoritative tone than the winning answer, and almost seems to be chiding some of the other answers for their specificity), ``METAL GATE.'' Perhaps the difference between the winning and losing proposals here is that the level of detail (even if false) and language used in the winning post suggests a higher level of authority in the image than the losing proposals.

How to account theoretically for more realistic worker beliefs is a subject of future research, but we sketch some ideas here (we also explore this question experimentally in the next section). We can for instance assume a belief model where the worker at time $t$ has beliefs above the relative likelihood of each proposal succeeding so far. We can represent those beliefs as real numbers $w_{m_0}, \ldots, w_{m_t} \in [0,1]$, where $t_m$ is the number of proposals at time t $m_t$. Let $w' \in [0,1]$ represent the worker's confidence in his or her own choice. Finally let $P_t$ be a distribution over $[0,1]$ from which future worker's submissions are drawn. The likelihood that the if the worker votes he or she chooses the contribution proposed at time $t' < t$ is $m_t'/\sum_{0 \leq i \leq w_{m_t}} w_i$. Under the zero knowledge assumption, each of these weights is identical. 

\section{Conclusion}
In this paper we discussed and formalized a model for eliciting varying levels of diversity in contributions from crowd workers. While the individual components of this contribution are built on common building blocks. The tradeoff in response set size that we have demonstrated has not been demonstrated before to our knowledge. Our results suggest future work in exploring why and in what circumstances workers will self-select out of a task, and provides a basis for exploring mechanism-based self-filtering of responses by workers.

\newpage

\bibliographystyle{alpha}
\bibliography{tuning}

\end{document}